\documentclass{toc}

\tocdetails{%
volume=8, number=23, year=2012, firstpage=513,
received={August 26, 2011},
published={September 25, 2012},
doi={10.4086/toc.2012.v008a023}
}

\newcommand\calB{{\cal B}}
\newcommand\calC{{\cal C}}
\newcommand\calL{{\cal L}}

\newcommand{\tint}{\operatorname{int}}
\newcommand{\tBCH}{\operatorname{BCH}}

\newcommand{\SAT}{\mathsf{SAT}}
\newcommand{\SVP}{\mathsf{SVP}}

\newcommand{\CVP}{\mathsf{CVP}}

\newcommand{\GapSVPnorm}[2]{\mathsf{GapSVP}^{#1}_{#2}}

\newcommand{\NP}{\mathsf{NP}}
\newcommand{\RP}{\mathsf{RP}}
\newcommand{\RTIME}{\mathsf{RTIME}}
\newcommand{\RSUBEXP}{\mathsf{RSUBEXP}}
\renewcommand{\P}{\mathsf{P}}

\newcommand{\DTIME}{\mathsf{DTIME}}
\newcommand{\YES}{\mathsf{YES}}
\newcommand{\NO}{\mathsf{NO}}

\newcommand{\eps}{\epsilon}
\renewcommand{\epsilon}{\varepsilon}

\newcommand{\poly}{\mathop{\mathrm{poly}}}
\newcommand{\trace}{\mathop{\mathrm{tr}}}
\newcommand{\linspan}{\mathop{\mathrm{span}}}
\runningtitle{Tensor-based Hardness of the Shortest Vector Problem}

\tocpdftitle{Tensor-based Hardness of the Shortest Vector Problem to within Almost Polynomial Factors}

\runningauthor{Ishay Haviv and Oded Regev}
\copyrightauthor{Ishay Haviv and Oded Regev}
\tocpdfauthor{Ishay Haviv, Oded Regev}

\begin{document}

\begin{frontmatter}[classification=text]
\title{Tensor-based Hardness of the\\ Shortest Vector Problem to within\\ Almost Polynomial Factors}

\author[ishay]{Ishay Haviv\thanks{Supported by the Adams Fellowship Program of the Israel
Academy of Sciences and Humanities. Work done while at Tel Aviv University.}}
\author[oded]{Oded Regev\thanks{Supported by an Alon Fellowship, by the Binational Science Foundation, by the Israel Science Foundation, by the European Commission under the Integrated Project QAP funded by the IST directorate as Contract Number 015848, and by the European Research Council (ERC) Starting Grant.}}

\begin{abstract}
We show that unless $\NP \subseteq \RTIME (2^{\poly(\log{n})})$, there is no polynomial-time algorithm approximating the Shortest Vector Problem ($\SVP$) on $n$-dimensional lattices in the $\ell_p$ norm ($1 \leq p<\infty$) to within a factor of
$2^{(\log{n})^{1-\eps}}$ for any $\eps >0$. This improves the previous best
factor of $2^{(\log{n})^{1/2-\eps}}$ under the same complexity assumption due
to Khot (J.~ACM,~2005). Under the stronger assumption $\NP \nsubseteq \RSUBEXP$, we obtain a hardness factor of $n^{c/\log\log{n}}$ for some $c>0$.

Our proof starts with Khot's $\SVP$ instances that are hard
to approximate to within some constant. To boost the hardness factor we simply apply the standard tensor product of lattices. The main novelty is in the analysis, where we show that the lattices of Khot behave
nicely under tensorization. At the heart of the analysis is a
certain matrix inequality which was first used in the context
of lattices by de Shalit and Parzanchevski (2006).
\end{abstract}

\tocacm{F.2.2, F.1.3, G.1.6}
\tocams{68Q17, 52C07, 11H06, 11H31, 05B40}

\tockeywords{lattices, shortest vector problem, NP-hardness, hardness of approximation}

\end{frontmatter}

\section{Introduction}

A \emph{lattice} is a periodic geometric object defined as the set of all integer
combinations of some linearly independent vectors in $\R^n$. The interesting
combinatorial structure of lattices has been investigated by mathematicians over the
last two centuries, and for at least three decades it has also been studied from an asymptotic
algorithmic point of view. Roughly speaking, most fundamental
problems on lattices are not known to be efficiently solvable. Moreover, there
are hardness results showing that such problems cannot be solved by
polynomial-time algorithms unless the polynomial-time hierarchy collapses. One of
the main motivations for research on the hardness of lattice problems is their
applications in cryptography, as was demonstrated by
Ajtai~\cite{ajtai96generating}, who came up with a construction of
cryptographic primitives whose security relies on the worst-case hardness of
certain lattice problems.

Two main computational problems associated with lattices are the
Shortest Vector Problem ($\SVP$) and the Closest Vector Problem
($\CVP$). In the former, for a lattice given by \emph{some} basis we
are supposed to find (the length of) a shortest nonzero vector in
the lattice. The problem $\CVP$ is an inhomogeneous variant of
$\SVP$, in which given a lattice and some target point one has to
find (its distance from) the closest lattice point. The hardness of
lattice problems partly comes from the fact that there are many
possible bases for the same lattice.

In this paper we improve the best hardness result known for $\SVP$. Before
presenting our results let us start with an overview of related work.

\subsection{Related work}

In the early 1980s, Lenstra, Lenstra, and Lov\'asz (LLL)~\cite{LLL82} presented the first
polynomial-time approximation algorithm for $\SVP$. Their
algorithm achieves an approximation factor of $2^{O(n)}$, where $n$ is the
dimension of the lattice. Using their algorithm, Babai~\cite{Babai86} gave an approximation
algorithm for $\CVP$ achieving the same approximation factor. A
few years later, improved algorithms were presented for both problems,
obtaining a slightly sub-exponential approximation factor, namely
$2^{O(n(\log\log{n})^2 / \log{n})}$~\cite{SchnorrBKZ}, and this has since been
improved slightly~\cite{AjtaiKS01,MicV10}.
The best algorithm known for solving
$\SVP$ \emph{exactly} requires exponential running time in
$n$~\cite{sep:Kannan:87,AjtaiKS01,MicV10}. All the above results hold with respect to
any $\ell_p$ norm ($1 \leq p \leq \infty$).

On the hardness side, it was proven in 1981 by van Emde Boas~\cite{vEB81} that it is
$\NP$-hard to solve $\SVP$ exactly in the $\ell_\infty$ norm. The
question of extending this result to other norms, and in particular to the
Euclidean norm $\ell_2$, remained open until the breakthrough result by Ajtai~\cite{Ajtai98}
showing that exact $\SVP$ in the $\ell_2$ norm is $\NP$-hard under randomized
reductions. Then, Cai and Nerurkar~\cite{CaiNerurkar99} obtained hardness of
approximation to within $1+n^{-\eps}$ for any $\eps>0$.
The first inapproximability result of $\SVP$ to within a factor bounded
away from $1$ is that of Micciancio~\cite{Mic01svp}, who showed that under randomized
reductions $\SVP$ in the $\ell_p$ norm is $\NP$-hard to approximate to within
any factor smaller than $\sqrt[p]{2}$. For the $\ell_\infty$
norm, a considerably stronger result is known:
Dinur~\cite{Dinur00} showed that $\SVP$ is $\NP$-hard to
approximate in the $\ell_\infty$ norm to within a factor of
$n^{c/\log\log{n}}$ for some constant $c>0$.

To date, the strongest hardness result known for $\SVP$ in the $\ell_p$
norm is due to Khot~\cite{Khot05svp} who showed $\NP$-hardness of approximation
to within arbitrarily large constants under randomized reductions for any $1 < p
<\infty$. Furthermore, under randomized quasipolynomial-time reductions (\ie,
reductions that run in time $2^{\poly(\log{n})}$), the hardness factor becomes
$2^{(\log{n})^{1/2-\eps}}$ for any $\eps>0$. Khot speculated there that it might be possible to improve this to $2^{(\log{n})^{1-\eps}}$, as this is the hardness factor known for the analogous problem in linear codes~\cite{dumer99hardness}.

Khot's proof does not work for the $\ell_1$ norm. However, it was shown in~\cite{RegevR06}
that for lattice problems, the $\ell_2$ norm is the easiest
in the following sense: for any $1\le p\le \infty$, there exists a randomized
reduction from lattice problems such as $\SVP$ and $\CVP$ in the $\ell_2$ norm
to the respective problem in the $\ell_p$ norm with essentially the same approximation factor.
In particular, this implies that Khot's results also hold for the $\ell_1$ norm.

Finally, we mention that a considerably stronger result is known for $\CVP$,
namely that for any $1\le p\le \infty$, it is $\NP$-hard to approximate $\CVP$
in the $\ell_p$ norm to within $n^{c/\log\log{n}}$ for some constant $c>0$~\cite{DinurKRS03}.
We also mention that in contrast to the above hardness results, it is known that for any $c>0$, $\SVP$ and $\CVP$
are unlikely to be $\NP$-hard to approximate to within a $\sqrt{c n/\log n}$ factor, as this
would imply the collapse of the polynomial-time
hierarchy~\cite{Goldreich:1998:LNA,AharonovR04}.

\subsection{Our results}

The main result of this paper improves the best $\NP$-hardness factor known for
$\SVP$ under randomized quasipolynomial-time reductions. This and two additional
hardness results are stated in the following theorem. Here, $\RTIME$ is the randomized one-sided
error analogue of $\DTIME$. Namely, for a function $f$ we denote by
$\RTIME(f(n))$ the class of problems having a probabilistic algorithm running in time $O(f(n))$ on inputs of size $n$ that accepts $\YES$ inputs
with probability at least $2/3$, and rejects $\NO$ inputs with certainty.

\begin{theorem}\label{thm:Intro}
For every $1 \leq p \leq \infty$ the following holds.
\begin{enumerate}
    \item \label{itm:intro_1}
For every constant $c\ge 1$, there is no polynomial-time
algorithm that approximates $\SVP$ in the $\ell_p$ norm to within a
factor of $c$ unless 
\[
\NP \subseteq \RP = \bigcup_{c \ge 1} \RTIME (n^c)\,.
\] 
    \item \label{itm:intro_2}
For every $\eps>0$, there is no polynomial-time
algorithm that approximates $\SVP$ on $n$-dimensional lattices in the $\ell_p$ norm to within
    a factor of $2^{(\log{n})^{1-\eps}}$ unless 
\[
\NP \subseteq
\RTIME (2^{\poly(\log{n})})\,.
\]
    \item \label{itm:intro_3}
There exists a $c>0$ such that there is no polynomial-time
algorithm that approximates $\SVP$ on $n$-dimensional lattices in the $\ell_p$ norm to within
    a factor of $n^{c/\log\log{n}}$ unless 
\[
\NP \subseteq
\RSUBEXP = \bigcap_{\delta>0} \RTIME (2^{n^\delta})\,.
\]  %
\end{enumerate}
\end{theorem}

\expref{Theorem}{thm:Intro} improves on the best known hardness result for any $p <
\infty$. For $p=\infty$, a better hardness result is already known, namely that
for some $c>0$, approximating to within $n^{c/\log\log{n}}$ is $\NP$-hard~\cite{Dinur00}.
Moreover, \expref{item}{itm:intro_1} was already proved by Khot~\cite{Khot05svp} and we provide an alternative proof.
We remark that all three items follow from a more general statement (see \expref{Theorem}{thm:IntroGeneral}).

\subsection{Techniques}

A standard method to prove hardness of approximation for large constant or
super-constant factors is to first prove hardness for some fixed constant
factor, and then \emph{amplify} the constant using some polynomial-time (or
quasipolynomial-time) transformation. For example, the \emph{tensor product} of
linear codes is used to amplify the $\NP$-hardness of approximating the minimum
distance in a linear code of block length $n$ to arbitrarily large constants
under polynomial-time reductions and to $2^{(\log{n})^{1-\eps}}$ (for any
$\eps>0$) under quasipolynomial-time reductions~\cite{dumer99hardness}. This
example motivates one to use the tensor product of lattices to increase the
hardness factor known for approximating $\SVP$. However, whereas the minimum
distance of the $k$-fold tensor product of a code $\calC$ is simply the $k$th
power of the minimum distance of $\calC$, the behavior of the length of a
shortest nonzero vector in a tensor product of lattices is more complicated and
not so well understood.

Khot's approach in~\cite{Khot05svp} was to prove a constant hardness factor for
$\SVP$ instances that have some ``code-like'' properties. The rationale is that
such lattices might behave in a more predictable way under the tensor product.
The construction of these ``basic'' $\SVP$ instances is ingenious, and is based
on BCH codes as well as a restriction into a random sublattice.
However, even for these code-like lattices, the behavior of the tensor product
was not clear. To resolve this issue, Khot introduced
a variant of the tensor product, which he called \emph{augmented tensor product}, and
using it he showed the hardness factor of $2^{(\log{n})^{1/2-\eps}}$. This
unusual hardness factor can be seen as a result of the augmented tensor
product. In more detail, for the augmented tensor product to work, the dimension of Khot's basic
$\SVP$ instances grows to $n^{\Theta(k)}$, where $k$ denotes the number of times
we intend to apply the augmented tensor product. After
applying it, the dimension grows to $n^{\Theta(k^2)}$ and
the hardness factor becomes $2^{\Theta(k)}$. This limits the hardness factor as a
function of the dimension $n$ to $2^{(\log{n})^{1/2-\eps}}$.

Our main contribution is showing that Khot's basic $\SVP$ instances
\emph{do} behave well under the (standard) tensor product.
The proof of this fact uses a new method to analyze vectors in the tensor product of
lattices, and is related to a technique used by de Shalit and Parzanchevski~\cite{DeShalit06}.
\expref{Theorem}{thm:Intro} now follows easily: we start
with (a minor modification of) Khot's basic $\SVP$ instances, which
are known to be hard to approximate to within some constant.
We then apply the $k$-fold tensor product for appropriately chosen
values of $k$ and obtain instances of dimension $n^{O(k)}$ with hardness
$2^{\Omega(k)}$.

\subsection{Open questions}

Some open problems remain. The most obvious is proving that $\SVP$ is hard to
approximate to within factors greater than $n^{c/\log\log{n}}$
under some plausible complexity assumption. Such a result, however,
is not known for $\CVP$ nor for the minimum distance problem in linear codes,
and most likely proving it there first would be easier.
An alternative goal is to improve on the $O(\sqrt{n/\log n})$ upper bound beyond which $\SVP$
is not believed to be $\NP$-hard~\cite{Goldreich:1998:LNA,AharonovR04}.

A second open question is whether our complexity assumptions can be weakened.
For instance, our $n^{c/\log\log{n}}$ hardness result is based on the assumption
that $\NP \nsubseteq \RSUBEXP$. For $\CVP$, such a hardness factor is known based
solely on the assumption $\P \neq \NP$~\cite{DinurKRS03}. Showing something similar for
$\SVP$ would be very interesting.
In fact, coming up with a \emph{deterministic} reduction (even for constant
approximation factors) already seems very challenging; all known hardness
proofs for $\SVP$ in $\ell_p$ norms, $p<\infty$, use randomized
reductions. (We note, though, that~\cite{Mic01svp} does describe a deterministic reduction
based on a certain number-theoretic conjecture.) Ideas appearing in the recent
$\NP$-hardness proofs of the minimum distance problem in linear codes~\cite{ChengWan,AustrinKhot}
might be useful. Finally, we mention that a significant step towards derandomization
was recently made by Micciancio~\cite{MicSVP12}: he strengthened our results
by showing reductions with only one-sided error.

\subsection{Outline}

The rest of the paper is organized as follows. In
\expref{Section}{sec:preliminaries} we gather some background on lattices and on
the central tool in this paper---the tensor product of lattices. In
\expref{Section}{sec:proof} we prove \expref{Theorem}{thm:Intro}. For the sake of completeness,
\expref{Section}{appendix:Khot} provides a summary of Khot's
work~\cite{Khot05svp} together with the minor modifications that
we need to introduce.

\section{Preliminaries}\label{sec:preliminaries}

\subsection{Lattices}

A \emph{lattice} is a discrete additive subgroup of $\mathbb{R}^n$.
Equivalently, it is the set of all integer combinations
$$\calL(b_1,\ldots,b_m)=\left\{\sum^{m}_{i=1}{x_i \,b_i}:x_i \in \mathbb{Z}\mbox{ for all }1\leq i\leq m\right\}$$ of $m$
linearly independent vectors $b_1,\ldots , b_m$ in $\mathbb{R}^n$ $(n \geq m)$.
If the rank $m$ equals the dimension $n$, then we say that the lattice is \emph{full-rank}. The set $\{b_1,\ldots , b_m\}$ is called a \emph{basis} of the
lattice. Note that a lattice has many possible bases. We often represent a
basis by an $n \times m$ matrix $B$ having the basis vectors as columns, and we
say that the basis $B$ \emph{generates} the lattice $\calL$. In such case we
write $\calL = \calL(B)$. It is well known and easy to verify that
two bases $B_1$ and $B_2$ generate the same lattice if and only if $B_1 = B_2
U$ for some \emph{unimodular} matrix $U \in\Z^{m \times m}$ (\ie, a matrix
whose entries are all integers and whose determinant is $\pm 1$). The \emph{determinant} of a lattice generated by a basis $B$ is $\det(\calL(B)) =
\sqrt{\det{(B^{T}B)}}$. It is easy to show that the determinant of a lattice is
independent of the choice of basis and is thus well-defined. A \emph{sublattice} of $\calL$ is a lattice $\calL(S) \subseteq \calL$ generated by
some linearly independent lattice vectors $S=\{s_1,\ldots,s_r\}\subseteq
\calL$. It is known that any integer matrix $B$ can be written as $[H \;
0]U$ where $H$ has full column rank and $U$ is unimodular. One way to achieve
this is by using the Hermite Normal Form (see, \eg, \cite[Page
67]{CohenCourse}).

For any $1 \leq p < \infty$, the $\ell_p$ norm of a vector $x \in \R^n$ is
defined as $\|x\|_p=\sqrt[p]{\sum_i{|x_i|^p}}$ and its $\ell_\infty$ norm is
$\|x\|_\infty = \max_i{|x_i|}$. One basic parameter of a lattice $\calL$, denoted by
$\lambda_1^{(p)}(\calL)$, is the $\ell_p$ norm of a shortest nonzero vector in
it. Equivalently, $\lambda_1^{(p)}(\calL)$ is the minimum $\ell_p$ distance
between two distinct points in the lattice $\calL$. This definition can be
generalized to define the $i$th \emph{successive minimum} as the smallest $r$
such that $\calB_p(r)$ contains $i$ linearly independent lattice points, where
$\calB_p(r)$ denotes the $\ell_p$ ball of radius $r$ centered at the origin.
More formally, for any $1\le p\le \infty$, we define
\[
\lambda_i^{(p)}(\calL) = \min\Bigl\{r: \dim\Bigl(\linspan\bigl(\calL \cap
\calB_p(r)\bigr)\Bigr) \geq i\Bigr\}\,.
\]
We often omit the superscript in $\lambda_i^{(p)}$ when $p=2$.

In 1896, Hermann Minkowski~\cite{mink:1896} proved the following classical
result, known as Minkowski's First Theorem. We consider here the
$\ell_2$ norm, although the result has an easy extension to other norms. For
a simple proof the reader is referred to~\cite[Chapter 1, Section 1.3]{MicciancioBook}.

\begin{theorem}[Minkowski's First Theorem]\label{thm:mink_first}
For any rank-$r$ lattice $\calL$,
$$\det(\calL) \geq \left(\frac{\lambda_1(\calL)}{\sqrt{r}}\right)^r.$$
\end{theorem}

Our hardness of approximation results will be shown through the promise version
$\GapSVPnorm{p}{\gamma}$, defined for any $1\leq p \le \infty$ and for any
approximation factor $\gamma \leq 1$ as follows.

\begin{definition}[Shortest Vector Problem]\label{def:SVP}
An instance of $\GapSVPnorm{p}{\gamma}$ is a pair $(B,s)$, where $B$ is a
lattice basis and $s$ is a number. In $\YES$ instances
$\lambda_1^{(p)}(\calL(B)) \leq \gamma \cdot s$, and in $\NO$ instances
$\lambda_1^{(p)}(\calL(B)) > s$.
\end{definition}

\subsection{Tensor product of lattices}

A central tool in the proof of our results is the \emph{tensor product} of lattices. Let us first recall some basic definitions. For two column vectors $u$ and $v$ of
dimensions $n_1$ and $n_2$ respectively, we define their tensor product
$u\otimes v$ as the $n_1 n_2$-dimensional column vector
$$
\left( \begin{array}{c}
  u_1 v \\
  \vdots \\
  u_{n_1} v \\
\end{array}\right)\,.
$$
If we think of the coordinates of $u\otimes v$ as arranged in an $n_1 \times
n_2$ matrix, we obtain the equivalent description of $u \otimes v$ as the
matrix $u \cdot v^T$. More generally, any $n_1 n_2$-dimensional vector $w$ can be
written as an $n_1 \times n_2$ matrix $W$. To illustrate the use of this
notation, notice that if $W$ is the matrix corresponding to $w$ then
\begin{equation}\label{eq:normastrace}
\|w\|_2^2 = \trace(W \,W^T)\,.
\end{equation}
Finally, for an $n_1 \times m_1$ matrix $A$ and an $n_2 \times m_2$ matrix $B$,
one defines their tensor product $A \otimes B$ as the $n_1 n_2 \times m_1 m_2$
matrix
$$
\left( \begin{array}{ccc}
  A_{11} B & \cdots & A_{1 m_1} B \\
  \vdots && \vdots \\
  A_{n_1 1} B & \cdots & A_{n_1 m_1} B \\
\end{array}\right)\,.
$$

Let $\calL_1$ be a lattice generated by the $n_1 \times m_1$ matrix $B_1$ and
$\calL_2$ be a lattice generated by the $n_2 \times m_2$ matrix $B_2$. Then the
tensor product of $\calL_1$ and $\calL_2$ is defined as the $n_1
n_2$-dimensional lattice generated by the $n_1 n_2 \times m_1 m_2$ matrix $B_1
\otimes B_2$, and is denoted by $\calL = \calL_1 \otimes \calL_2$. Equivalently,
$\calL$ is generated by the $m_1 m_2$ vectors obtained by taking the tensor product of
two column vectors, one from $B_1$ and one from $B_2$. If we think of the
vectors in $\calL$ as $n_1 \times n_2$ matrices, then we can also define it as
$$\calL = \calL_1 \otimes \calL_2 = \{B_1 X B_2^T : X \in \Z^{m_1 \times m_2}\}\,,$$
with each entry in $X$ corresponding to one of the $m_1 m_2$ generating
vectors. We will mainly use this definition in the proof of the main result.

As alluded to before, in the present paper we are interested in the behavior of
the shortest nonzero vector in a tensor product of lattices. It is easy to see that for
any $1\le p\le \infty$ and any two lattices $\calL_1$ and $\calL_2$, we have
\begin{eqnarray}\label{eq:TensorLeq}
    \lambda_1^{(p)}(\calL_1 \otimes \calL_2) \leq \lambda_1^{(p)}(\calL_1) \cdot \lambda_1^{(p)}(\calL_2).
\end{eqnarray}
Indeed, any two vectors $v_1$ and $v_2$ satisfy $\|v_1 \otimes v_2\|_p = \|v_1\|_p \cdot \|v_2\|_p$. Applying this to shortest nonzero vectors of $\calL_1$ and $\calL_2$ implies inequality~\eqref{eq:TensorLeq}.

Inequality~\eqref{eq:TensorLeq} has an analogue for linear codes, with
$\lambda_1^{(p)}$ replaced by the minimum distance of the code under the Hamming metric. There, it is not too
hard to show that the inequality is in fact an equality: the minimal distance
of the tensor product of two linear codes always equals to the product of their
minimal distances. However, contrary to what one might expect, there exist
lattices for which inequality~\eqref{eq:TensorLeq} is \emph{strict}. More precisely, for any sufficiently large $n$ there exist $n$-dimensional lattices $\calL_1$ and $\calL_2$ satisfying
$$\lambda_1(\calL_1 \otimes \calL_2) < \lambda_1(\calL_1) \cdot \lambda_1(\calL_2)\,.$$
The following lemma due to Steinberg shows this fact. Although we do not use this fact later on, the proof is
instructive and helps motivate the need for a careful analysis of tensor
products. To present this proof we need the notion of a \emph{dual lattice}. For
a full-rank lattice $\calL \subseteq \R^n$, its \emph{dual} lattice $\calL^{*}$
is defined as
$$\calL^{*} = \{x \in \R^n : \langle x,y \rangle\in\Z \mbox{ for
all $y \in \calL$}\}\,.$$ A \emph{self-dual} lattice is one that satisfies
$\calL = \calL^{*}$. It can be seen that for a full-rank lattice $\calL$
generated by a basis $B$, the basis $(B^{-1})^T$ generates the lattice
$\calL^{*}$.

\begin{lemma}[{\cite[Page 48]{MilHus73}}]\label{lem:tensor}
For any $n \ge 1$ there exists an $n$-dimensional self-dual lattice
$\calL$ satisfying $\lambda_1(\calL \otimes \calL^*) \leq \sqrt{n}$ and
$\lambda_1(\calL) = \lambda_1(\calL^*) = \Omega(\sqrt{n})$.
\end{lemma}

\begin{proof}
We first show that for any full-rank $n$-dimensional lattice $\calL$,
$\lambda_1(\calL \otimes \calL^*) \le \sqrt{n}$. Let $\calL$ be a lattice
generated by a basis $B=(b_1,\ldots,b_n)$. Let
$(B^{-1})^T=(\tilde{b_1},\ldots,\tilde{b_n})$ be the basis generating its
dual lattice $\calL^{*}$. Now consider the vector $\sum_{i=1}^{n}{b_i \otimes
\tilde{b_i}} \in \calL \otimes \calL^*$. Using our matrix notation, this vector
can be written as
$$B \,I_n ((B^{-1})^T)^T = B \,B^{-1} = I_n\,,$$
and clearly has $\ell_2$ norm $\sqrt{n}$.
To complete the proof, we need to use the (non-trivial) fact that for any
$n \ge 1$ there exists a full-rank, $n$-dimensional and self-dual lattice with
shortest nonzero vector of norm $\Omega(\sqrt{n})$. This fact is due to Conway
and Thompson; see~\cite[Page 46]{MilHus73} for details.
\end{proof}

\section{Proof of results}\label{sec:proof}

The following is our main technical result. As we will show later, \expref{Theorem}{thm:Intro} follows easily by plugging in appropriate values of $k$.

\begin{theorem}\label{thm:IntroGeneral}
For any $1 \leq p \leq \infty$ there exist $c,C>0$ such that the following holds.
There exists a randomized reduction that takes as input a $\SAT$ instance and an integer $k \ge 1$ and outputs
a $\GapSVPnorm{p}{\gamma}$ instance of dimension $n^{Ck}$ with gap $\gamma = 2^{-ck}$, where $n$ denotes the size of the $\SAT$ instance.
The reduction runs in time polynomial in $n^{Ck}$ and has two-sided error, namely,
given a $\YES$ (resp., $\NO$) instance it outputs a $\YES$ (resp., $\NO$) instance with probability $9/10$.
\end{theorem}

In fact, we will only need to prove this theorem for the case $p=2$ since, as is easy to see, the general case follows from the following theorem (applied with, say, $\eps=1/2$).\footnote{We note that our results can be shown directly for any $1 <p < \infty$ without using \expref{Theorem}{thm:RosenRegev} by
essentially the same proof.}

\begin{theorem}[\cite{RegevR06}]\label{thm:RosenRegev}
For any $\eps >0$, $\gamma < 1$ and $1 \leq p \leq \infty$ there
exists a randomized polynomial-time reduction from
$\GapSVPnorm{2}{\gamma'}$ to $\GapSVPnorm{p}{\gamma}$, where $\gamma'
= (1-\eps)\gamma$.
\end{theorem}

\subsection{Basic \texorpdfstring{$\SVP$}{SVP}}

As already mentioned, our reduction is crucially based on a hardness result of
a variant of $\SVP$ stemming from Khot's work~\cite{Khot05svp}.
Instances of this variant have properties that make it possible to amplify
the gap using the tensor product.
The following theorem summarizes the hardness result on which our proof is based.
For a proof the reader is referred to \expref{Section}{appendix:Khot}.

\begin{theorem}[\cite{Khot05svp}]\label{thm:Khot}
There are a constant $\gamma <1$ and a polynomial-time randomized
reduction from $\SAT$ to $\SVP$ outputting a lattice basis $B$,
satisfying $\calL(B) \subseteq \Z^n$ for some integer $n$, and an
integer $d$, such that:
\begin{enumerate}
    \item For any $\YES$ instance of $\SAT$, with probability at least $9/10$, $\lambda_1 (\calL(B)) \leq \gamma \cdot \sqrt{d}$.
    \item For any $\NO$ instance of $\SAT$, with probability at least $9/10$, for every nonzero vector $v \in \calL(B)$,
\begin{itemize}
    \item $v$ has at least $d$ nonzero coordinates, or
    \item all coordinates of $v$ are even and at least $d/4$ of them are nonzero, or
    \item all coordinates of $v$ are even and $\|v\|_2 \geq d$.
\end{itemize}
In particular, $\lambda_1 (\calL(B)) \geq \sqrt{d}$.
\end{enumerate}
\end{theorem}

\subsection{Boosting the \texorpdfstring{$\SVP$}{SVP} hardness factor}

As mentioned before, we boost the hardness factor using the tensor
product of lattices. For a lattice $\calL$ we denote by $\calL^{\otimes k}$ the
$k$-fold tensor product of $\calL$. An immediate corollary of inequality
\eqref{eq:TensorLeq} is that if $(B,d)$ is a $\YES$ instance of the $\SVP$
variant in \expref{Theorem}{thm:Khot}, and $\calL=\calL(B)$, then
\begin{eqnarray}\label{eq:YESinstances}
    \lambda_1(\calL^{\otimes k}) \leq \gamma ^k d^{k/2}.
\end{eqnarray}
For the case in which $(B,d)$ is a $\NO$ instance we will show that any nonzero
vector of $\calL^{\otimes k}$ has norm at least $d^{k/2}$, \ie,
\begin{eqnarray}\label{eq:NOinstances}
    \lambda_1(\calL^{\otimes k}) \geq d^{k/2}.
\end{eqnarray}
This yields a gap of $\gamma^k$ between the two cases. Inequality
\eqref{eq:NOinstances} easily follows by induction from the central
lemma below, which shows that $\NO$ instances ``tensor nicely.''

\begin{lemma}\label{lem:NOinstances}
Let $(B,d)$ be a $\NO$ instance of the $\SVP$ variant given in
\expref{Theorem}{thm:Khot}, and denote by $\calL_1$ the lattice generated by the
basis $B$. Then for any lattice $\calL_2$,
$$\lambda_1(\calL_1 \otimes \calL_2) \geq \sqrt{d} \cdot \lambda_1(\calL_2)\,.$$
\end{lemma}

The proof of this lemma is based on some properties of sublattices of $\NO$
instances which are established in the following claim.

\begin{claim}\label{claim:sublattice}
Let $(B,d)$ be a $\NO$ instance of the $\SVP$ variant given in
\expref{Theorem}{thm:Khot}, and let $\calL \subseteq \calL(B)$ be a sublattice of rank $r>0$. Then at least one of the following properties holds:
\begin{enumerate}
    \item Every basis matrix of $\calL$ has at least $d$ nonzero rows (\ie, rows that are not all zero).
    \item Every basis matrix of $\calL$ contains only even entries and has at least $d/4$ nonzero rows.
    \item $\det(\calL) \geq d^{r/2}$.
\end{enumerate}
\end{claim}

\begin{proof}
Assume that $\calL$ does not have either of the first two properties. Our goal is
to show that the third property holds. Since the first property does
not hold, we have $r<d$ and also that any vector in $\calL$ has fewer than $d$ nonzero
coordinates. By \expref{Theorem}{thm:Khot}, this implies that $\calL \subseteq 2\cdot \Z^n$. By the assumption that
the second property does not hold, there must exist a basis of $\calL$ that has fewer than
$d/4$ nonzero rows. Therefore, all nonzero vectors in $\calL$ have
fewer than $d/4$ nonzero coordinates, and hence have norm at least $d$, again by \expref{Theorem}{thm:Khot}.
We conclude that $\lambda_1(\calL) \ge d$, and by Minkowski's First Theorem
(\expref{Theorem}{thm:mink_first}) and $r < d$ we have
\[
\det(\calL) \geq \left(\frac{\lambda_1(\calL)}{\sqrt{r}}\right)^{r} \geq d^{r/2}.
\qedhere \]
\end{proof}

\begin{proof}[Proof of \expref{Lemma}{lem:NOinstances}]
Let $v$ be an arbitrary nonzero vector in $\calL_1 \otimes \calL_2$. Our goal is to
show that $\|v\|_2 \geq \sqrt{d} \cdot \lambda_1(\calL_2)$. We can write $v$ in
matrix notation as $B_1 X {B_2}^T$, where the integer matrix $B_1$ is a basis of $\calL_1$, $B_2$ is a basis of $\calL_2$,
and $X$ is an integer matrix of coefficients. Let
$U$ be a unimodular matrix for which $X = [H \; 0] U$, where $H$ is a matrix
with full column rank. Thus, the vector $v$ can be written as $B_1 [H \;
 0] (B_2 U^T)^T$. Since $U^T$ is also unimodular, the matrices $B_2$ and $B_2 U^T$ generate the same lattice. Now remove from $B_2 U^T$ the columns corresponding to the zero
columns in $[H \; 0]$ and denote the resulting matrix by $B'_2$. Furthermore, denote the
matrix $B_1 H$ by $B'_1$. Observe that both of the matrices $B'_1$ and $B'_2$ are
bases of the lattices they generate, \ie, they have full column rank. The vector
$v$ equals $B'_1 {B'_2}^T$, where $\calL'_1 := \calL(B'_1) \subseteq \calL_1$
and $\calL'_2 := \calL(B'_2) \subseteq \calL_2$.

\expref{Claim}{claim:sublattice} guarantees that the lattice $\calL'_1$ defined above has at least one of the three properties mentioned in the claim. We show that $\|v\|_2 \geq \sqrt{d} \cdot \lambda_1(\calL'_2)$ in each of these three cases. Then, by the fact that $\lambda_1(\calL'_2) \geq \lambda_1(\calL_2)$, the lemma will follow.
\begin{description}
    \item[Case 1:] Assume that at least $d$ of the rows in the basis matrix $B'_1$ are nonzero. Thus, at least $d$ of the rows of $B'_1 {B'_2}^T$ are nonzero lattice points from $\calL'_2$, and thus
    $$\|v\|_2 \geq \sqrt{d} \cdot \lambda_1(\calL'_2)\,.$$

    \item[Case 2:] Assume that the basis matrix $B'_1$ contains only even entries and has at least $d/4$ nonzero rows. Hence, at least $d/4$ of the rows of $B'_1 {B'_2}^T$ are even multiples of nonzero lattice vectors from $\calL'_2$. Therefore, every such row has $\ell_2$ norm at least $2 \cdot \lambda_1(\calL'_2)$, and it follows that
    $$\|v\|_2 \geq \sqrt{\frac{d}{4}} \cdot 2\cdot\lambda_1(\calL'_2) = \sqrt{d} \cdot \lambda_1(\calL'_2)\,.$$
\end{description}

The third case is based on the following central claim, which is similar to
Proposition 1.1 in~\cite{DeShalit06}. The proof is based on an elementary
matrix inequality relating the trace and the determinant of a symmetric
positive semidefinite matrix (see, \eg, \cite[Page 47]{bhatia:1997}).

\begin{claim}\label{claim:DeShalitType}
Let $\calL_1$ and $\calL_2$ be two rank-$r$ lattices generated by the bases
$U=(u_1, \ldots, u_r)$ and $W=(w_1, \ldots, w_r)$ respectively. Consider
the vector $v = \sum_{i=1}^{r}{u_i \otimes w_i}$ in $\calL_1 \otimes \calL_2$,
which can be written as $U I_r W^T = UW^T$ in matrix notation. Then,
$$\|v\|_2 \geq \sqrt{r} \cdot \bigl(\det(\calL_1) \cdot \det(\calL_2)\bigr)^{1/r}\,.$$
\end{claim}
\begin{proof}
Define the two $r\times r$ symmetric positive definite matrices $G_1 = U^T U$
and $G_2 = W^T W$ (known as the Gram matrices of $U$ and $W$). By the fact that
$\trace(AB)=\trace(BA)$ for any matrices $A$ and $B$ and
by equation~\eqref{eq:normastrace},
$$
\|v\|_2^2 =  \trace\bigl((UW^T)(UW^T)^T\bigr) = \trace(G_1 G_2) = \trace\bigl(G_1 G_2^{1/2}
G_2^{1/2} \bigr) = \trace\bigl(G_2^{1/2} G_1 G_2^{1/2} \bigr)\,,$$
where $G_2^{1/2}$ is the positive square root of $G_2$.
The matrix $G=G_2^{1/2} G_1 G_2^{1/2}$ is also symmetric and positive definite,
and as such it has $r$ real and positive eigenvalues. We can thus apply
the inequality of arithmetic and geometric means on these eigenvalues to get
\[
\|v\|_2^2 = \trace(G) \geq r \det(G)^{1/r} = r \cdot \bigl(\det(G_1) \cdot
\det(G_2)\bigr)^{1/r}\,.
\]
Taking the square root of both sides of this equation completes the proof.
\end{proof}

Equipped with \expref{Claim}{claim:DeShalitType} we turn to deal with the third
case. In order to bound from below the norm of $v$, we apply
the claim to its matrix form $B'_1 {B'_2}^T$ with the lattices $\calL'_1$
and $\calL'_2$ as above.

\begin{description}
\item[Case 3:] Assume that the lattice $\calL'_1$ satisfies $\det(\calL'_1) \geq d^{r/2}$, where $r$ denotes its rank. Combining \expref{Claim}{claim:DeShalitType} and Minkowski's First Theorem we have that
$$\|v\|_2 \geq \sqrt{r} \cdot \bigl(\det(\calL'_1) \cdot \det(\calL'_2)\bigr)^{1/r} \geq \sqrt{r} \cdot (d^{r/2})^{1/r} \cdot \frac{\lambda_1(\calL'_2)}{\sqrt{r}} = \sqrt{d} \cdot \lambda_1(\calL'_2)\,,$$
and this completes the proof of the lemma.\qedhere
\end{description}
\end{proof}

\subsection{Proof of the main theorem}

\begin{proof}[Proof of \expref{Theorem}{thm:IntroGeneral}]
Recall that it suffices to prove the theorem for $p=2$.
Given a $\SAT$ instance of size $n$, we apply the reduction from \expref{Theorem}{thm:Khot} and obtain in time $\poly(n)$ a pair
$(B,d)$ where $B$ is a basis of a $\poly(n)$-dimensional lattice. We then output $(B^{\otimes k},d^{k})$,
where $B^{\otimes k}$ is the $k$-fold tensor product of $B$, \ie, a basis of
the lattice ${\calL(B)}^{\otimes k}$. The dimension of this lattice
is $\poly(n^k)$, and combining inequalities \eqref{eq:YESinstances} and
\eqref{eq:NOinstances} we infer a gap of $2^{-ck}$.
\end{proof}

\begin{proof}[Proof of \expref{Theorem}{thm:Intro}]
For \expref{item}{itm:intro_1}, choose $k$ to be a sufficiently large
constant and apply \expref{Theorem}{thm:IntroGeneral}. This shows that any constant factor
approximation algorithm to $\SVP$ implies a two-sided error algorithm for $\SAT$.
Using known self-reducibility properties of $\SAT$ (see, \eg,~\cite[Chapter 11]{papadi}), this also implies a one-sided error
polynomial-time algorithm for $\SAT$. For \expref{item}{itm:intro_2}, apply \expref{Theorem}{thm:IntroGeneral} with $k = (\log{n})^{1/\eps}$ (where $n$ is the size of the input $\SAT$ instance) and let
$N=n^{C k}$ be the dimension of the output lattice. Since
$$k = \left(\frac{\log{N}}{C}\right)^{\frac{1}{1+\eps}} > \left(\frac{\log{N}}{C}\right)^{1-\eps},$$
the gap we obtain as a function of the dimension $N$ is $2^{\Omega((\log{N})^{1-\eps})}$.
Therefore, an algorithm that approximates $\SVP$ better than this gap implies a
randomized $\SAT$ algorithm running in time $2^{\poly(\log{n})}$,
and hence the desired containment $\NP \subseteq \RTIME(2^{\poly(\log{n})})$.
\expref{Item}{itm:intro_3} follows similarly by applying \expref{Theorem}{thm:IntroGeneral} with $k=n^{\delta}$ for all $\delta>0$.
\end{proof}

\section{Proof of \texorpdfstring{\expref{Theorem}{thm:Khot}}{Theorem 3.3}}\label{appendix:Khot}

In this section we prove \expref{Theorem}{thm:Khot}.
The proof is essentially the same as the one
in~\cite{Khot05svp} with minor modifications.

\subsection{Comparison with Khot's theorem}

For the reader familiar with Khot's proof, we now describe how
\expref{Theorem}{thm:Khot} differs from the one in~\cite{Khot05svp}. First, our
theorem is only stated for the $\ell_2$ norm (since we use \expref{Theorem}{thm:RosenRegev}
to extend the result to other norms). Second, the $\YES$
instances of Khot had another property that we do not need here (namely, that
the coefficient vector of the short lattice vector is also short). Third, as a
result of the augmented tensor product, Khot's theorem includes an extra
parameter $k$ that specifies the number of times the lattice is supposed to be
tensored with itself. Since we do not use the augmented tensor product, we
simply fix $k$ to be some constant. In more detail, we choose the number of
columns in the BCH code to be $d^{O(1)}$, as opposed to $d^{O(k)}$. This
eventually leads to our improved hardness factor. Finally, the third
possibility in our $\NO$ case is different from the one in Khot's theorem (which
says that there exists a coordinate with absolute value at least $d^{O(k)}$).
We note that coordinates with huge values are used several times in Khot's
construction in order to effectively restrict a lattice to a subspace.
We instead work directly with the restricted lattice, making the reduction somewhat cleaner.

\subsection{The proof}

The proof of \expref{Theorem}{thm:Khot} proceeds in three steps. In the first, a
variant of the Exact Set Cover problem, which is known to be $\NP$-hard, is reduced
to a gap variant of $\CVP$. In the second step we construct a basis $B_{\tint}$
of a lattice which, informally, contains many short vectors in the $\YES$ case,
and few short vectors in the $\NO$ case. Finally, in the third step we complete
the reduction by taking a random sublattice.

\subsubsection{Step 1}
First, consider the following variant of Exact Set Cover. Let $\eta>0$ be an
arbitrarily small constant. An instance of the problem is a pair $(S,d)$, where
$S = \{S_1,\ldots,S_{n'}\}$ is a collection of subsets of some universe
$[n'']=\{1,\ldots,n''\}$, and $d$ is a positive integer.
In $\YES$ instances, there exists $S' \subseteq S$ of
size $\eta d$ that covers each element of the universe exactly once. In $\NO$
instances, there is no $S' \subseteq S$ of size less than $d$ that covers all
elements of the universe. This problem is known to be $\NP$-hard for an arbitrarily small
$0 < \eta < 1$ and for $n' = O(d)$~\cite{BellareEtAl93}. Moreover, it is easy to see that the problem remains $\NP$-hard
if we fix $\eta$ to be any negative power of $2$ and restrict $d$ to be a power of $2$. Thus, to prove \expref{Theorem}{thm:Khot}, it suffices to reduce from this problem.

In the first step we use a well-known reduction from the above
variant of Exact Set Cover to a variant of $\CVP$. For an instance
$(S,d)$ we identify $S$ with the $n'' \times n'$ matrix over
$\{0,1\}$ whose columns are the characteristic vectors of the sets
in $S$. The reduction outputs an instance $(B_{\CVP},t)$, where
$B_{\CVP}$ is a basis generating the lattice $\{y \in \Z^{n'} :
Sy=0\}$ and $t$ is some integer vector satisfying
$St=-(1,1,\ldots,1)$. (If no such $t$ exists, the
reduction outputs an arbitrary $\NO$ instance.)
We note that given $S$ the basis $B_{\CVP}$ can be
constructed in polynomial time (see, \eg,~\cite[Lemma~3.1]{MicSODA08}).

\begin{lemma}\label{lem:B_cvp}%
If $(S,d)$ is a $\YES$ instance of the above variant of Exact Set
Cover, then there is a lattice vector $z \in \calL(B_{\CVP})$ such
that $z-t$ is a $\{0,1\}$ vector and has exactly $\eta d$
coordinates equal to $1$. If $(S,d)$ is a $\NO$ instance, then for
any lattice vector $z \in \calL(B_{\CVP})$ and any nonzero integer
$j_0$, the vector $z+j_0 t$ has at least $d$ nonzero coordinates.
\end{lemma}

\begin{proof}
If $(S,d)$ is a $\YES$ instance then there exists a vector $y \in
\{0,1\}^{n'}$ with exactly $\eta d$ coordinates equal to $1$ for
which $Sy=(1,1,\ldots,1)$. This implies that $S(y+t)=0$, so $z=y+t$
is the required lattice vector. On the other hand, if $(S,d)$ is a
$\NO$ instance, then for any $z \in \calL(B_{\CVP})$ we have
$S(z+j_0 t) = -j_0 \cdot (1,1,\ldots,1)$. This implies that the
nonzero coordinates of $z+j_0 t$ correspond to a cover $S' \subseteq
S$ of all elements in $[n'']$, and hence their number must be at least
$d$.
\end{proof}

\subsubsection{Step 2}
The second step of the reduction is based on BCH codes, as described in the
following theorem.

\begin{theorem}[{\cite[Page 255]{Alon00}}]\label{thm:Alon}
Let $N,d,h$ be integers satisfying $h=({d}/{2}) \log_2{N}$. Then
there exists an efficiently constructible matrix $P_{\tBCH}$ of size
$h \times N$ with $\{0,1\}$ entries such that the rows of the matrix
are linearly independent over $GF(2)$ and any $d$ columns of the
matrix are linearly independent over $GF(2)$.
\end{theorem}

Let $B_{\tBCH}$ be a basis of the lattice $\{y \in \Z^N :
(P_{\tBCH})y \equiv 0 \pmod 2\}$. Such a basis can be
easily constructed in polynomial time by duality (see the preliminaries
in~\cite{MicciancioR08}). The next lemma states some properties of this lattice.

\begin{lemma}\label{lemma:L_BCH}
Every nonzero vector in $\calL(B_{\tBCH})$ either has at least $d$
nonzero coordinates or all of its coordinates are even. Also, for any $r \ge 1$
it is possible to find in polynomial time with
probability at least $99/100$ a vector $s \in \{0,1\}^{N}$, such that
there are at least 
\[
\frac{1}{100 \cdot 2^h} {\binom{N}{r}}
\]
distinct lattice vectors $z \in \calL(B_{\tBCH})$ satisfying that
$z-s$ is a $\{0,1\}$ vector with exactly $r$ coordinates equal to
$1$.
\end{lemma}

\begin{proof}
Let $y \in \calL(B_{\tBCH})$ be a nonzero lattice vector. Observe that
if $y$ has an odd coordinate then its odd coordinates correspond
to column vectors of $P_{\tBCH}$ that sum to the zero vector over
$GF(2)$. Therefore, their number must be at least $d$. This proves
the first statement.

We now prove the second statement. Consider the set $\calL(B_{\tBCH}) \cap \{0,1\}^N$ whose size is
$2^{N-h}$ (since as a subset of $GF(2)^N$ it is the kernel of $P_{\tBCH}$ whose dimension is $N-h$). In order to choose $s$ we first uniformly pick a
vector in this set and then we uniformly pick
$r$ of its coordinates and flip them. For a vector $s \in \{0,1\}^N$
let $A_s$ denote the number of ways in the above process to obtain
$s$ among the $2^{N-h} \cdot {\binom{N}{r}}$ possible ways.
The probability that the chosen $s$
satisfies 
\[
A_s \leq \frac{1}{100 \cdot 2^h} {\binom{N}{r}}\quad\text{is}\quad
\sum_{s \mbox{~s.t.~}A_s \leq \frac{1}{100 \cdot 2^h} {\binom{N}{r}}}{\frac{A_s}{2^{N-h}{\binom{N}{r}}}} \leq 2^N \cdot \frac{\frac{1}{100 \cdot 2^h}
{\binom{N}{r}}}{2^{N-h}{\binom{N}{r}}} \leq \frac{1}{100}\,,
\]
thus with probability at least $99/100$ we obtain an $s$
that satisfies 
\[
A_s > \frac{1}{100 \cdot 2^h} {\binom{N}{r}}\,.
\]
It remains to notice that such an $s$ also satisfies the requirement
in the statement of the lemma (since from each vector in $\{0,1\}^N$
of Hamming distance $r$ from $s$ we can obtain $z \in \calL(B_{\tBCH})$ as in
the statement by simply adding $2$ to a subset of its coordinates).
\end{proof}

We now construct the \emph{intermediate lattice} generated by a basis
matrix $B_{\tint}$ (see \expref{Figure}{fig:bint}). Let $\eta$ be a sufficiently small
constant, say $1/128$. Let $r = ({3}/{4}+\eta)d$
and choose $s$ as in \expref{Lemma}{lemma:L_BCH}. We choose the
parameters of $B_{\tBCH}$ to be $N = d^{2/\eta}$, $d$, and
$h=({d}/{2}) \log_2{N}$. Consider a matrix whose upper left block
is $2 \cdot B_{\CVP}$, whose lower right block is $B_{\tBCH}$, and whose
other entries are zeros. Adding to this matrix the column given by
the concatenation of $2 \cdot t$ and $s$, we obtain the basis matrix
$B_{\tint}$ of the intermediate lattice.

\begin{figure}[ht]
\begin{center}
\includegraphics[width=2in]{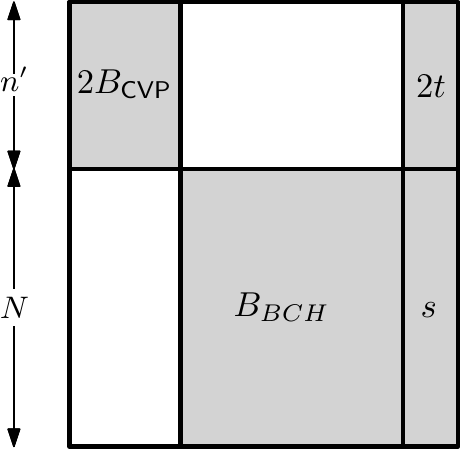}
\end{center}
\caption{$B_{\tint}$}
\label{fig:bint}
\end{figure}

The following two lemmas describe the properties of
$\calL(B_{\tint})$. The first one states that if the $\CVP$ instance is a
$\YES$ instance then $\calL(B_{\tint})$ contains many short vectors.
Define $\gamma = ({3}/{4}+5\eta)^{1/2} < 1$. A nonzero lattice
vector of $\calL(B_{\tint})$ is called \emph{good} if it has $\ell_2$
norm at most $\gamma \cdot \sqrt{d}$, has $\{0,1,2\}$ coordinates,
and has at least one coordinate equal to $1$.

\begin{lemma}\label{lemma:Good}%
If the $\CVP$ instance is a $\YES$ instance and the vector $s$
has the property from \expref{Lemma}{lemma:L_BCH}, then there are at least
$\frac{1}{100 \cdot 2^h} {\binom{N}{r}}$ good lattice vectors in
$\calL(B_{\tint})$.
\end{lemma}

\begin{proof}
Assume that the $\CVP$ instance is a $\YES$ instance. By
\expref{Lemma}{lem:B_cvp}, this implies that there exists $y$ such that
$B_{\CVP}y-t$ is a $\{0,1\}$ vector and has exactly $\eta d$
coordinates equal to $1$. Let $s$ be as in \expref{Lemma}{lemma:L_BCH}, so there are at
least $\frac{1}{100 \cdot 2^h} {\binom{N}{r}}$ distinct choices of
$x$ for which $(B_{\tBCH})x-s$ is a $\{0,1\}$ vector with exactly $r$
coordinates equal to $1$. For every such $x$, the lattice
vector\footnote{We use $\circ$ to denote concatenation of vectors.}
$$B_{\tint}(y \circ x \circ (-1)) = \bigl(2(B_{\CVP}y-t) \circ
((B_{\tBCH})x-s)\bigr)$$ has $\{0,1,2\}$ coordinates, has at least one
coordinate equal to $1$, and has $\ell_2$ norm $\sqrt{4\eta d + r} =
\gamma \cdot \sqrt{d}$, as required.
\end{proof}

The second lemma shows that if the $\CVP$ instance is a $\NO$ instance then
$\calL(B_{\tint})$ contains few vectors that do not have the property from
\expref{Theorem}{thm:Khot}, Item 2. We call such vectors \emph{annoying}. In more
detail, a lattice vector of $\calL(B_{\tint})$ is annoying if it satisfies \emph{all} of the following:
\begin{itemize}
    \item The number of its nonzero coordinates is smaller than $d$.
    \item Either it contains an odd coordinate or the number of its nonzero coordinates is smaller than $d/4$.
    \item Either it contains an odd coordinate or it has norm smaller than $d$.
\end{itemize}

\begin{lemma}\label{lemma:Annoying}
If the $\CVP$ instance is a $\NO$ instance, then there are at most
$d^{d/4} \cdot
{\binom{N+n'}{d/4}}$ annoying lattice vectors in
$\calL(B_{\tint})$.
\end{lemma}

\begin{proof}
Assume that the $\CVP$ instance is a $\NO$ instance and let
$B_{\tint}x$ be an annoying vector with coefficient vector $x = y
\circ z \circ (j_0)$. We have
$$B_{\tint}x = 2(B_{\CVP}y+j_0 t) \circ (B_{\tBCH}z+j_0 s)\,.$$
By \expref{Lemma}{lem:B_cvp}, if $j_0 \neq 0$ then the vector
$B_{\CVP}y+j_0 t$ has at least $d$ nonzero coordinates, so it is not
an annoying vector. Thus we can assume that $j_0 = 0$ and therefore
$B_{\tint}x = 2(B_{\CVP}y) \circ (B_{\tBCH}z)$.

Since $B_{\tint}x$ is annoying we know that it has fewer than $d$
nonzero coordinates, so by \expref{Lemma}{lemma:L_BCH} we get that all
coordinates of $B_{\tint}x$ are even. Again by the definition of an
annoying vector, we conclude that fewer than $d/4$ of the
coordinates of $B_{\tint}x$ are nonzero and all of them have absolute
value smaller than $d$. Thus, we get a bound of $d^{d/4} \cdot
{\binom{N+n'}{d/4}}$ on the number of possible choices for $B_{\tint}x$,
and this completes the proof of the lemma.
\end{proof}

\subsubsection{Step 3}
In the third step we construct the final $\SVP$ instance as claimed in
\expref{Theorem}{thm:Khot}. By \expref{Lemma}{lemma:Good}, the number of good vectors
in the $\YES$ case is at least
$$\frac{1}{100 \cdot 2^h}{\binom{N}{r}} = \frac{1}{100 \cdot 2^{(d \log_2{N}) / 2}}{\binom{N}{(3/4+\eta)d}} \geq \frac{N^{(3/4+\eta)d}}{100 \cdot d^d \cdot N^{d/2}} = \frac{N^{(1/4+\eta)d}}{100\cdot d^d}=:G\,.$$
By \expref{Lemma}{lemma:Annoying}, in the $\NO$ case there are at most $A:=d^{d/4} \cdot
{\binom{N+n'}{d/4}}$ annoying vectors.
By our choice of $N$ and the fact that $n' = O(d)$,
for sufficiently large $d$ we have $n' \leq N$ and hence
\[
A \leq  d^{d/4} \cdot (2N)^{d/4} \leq 10^{-5} \cdot G\,.
\]

Choose a prime $q$ in the interval $[100A,G/100]$ and let $w \in
\Z^{n'+N}$ be a vector whose coordinates are chosen randomly and
uniformly from the range $\{0,\ldots,q-1\}$. The final output of the
reduction is a basis $B$ of the lattice $\{x \in \calL(B_{\tint}) :
\langle w, x\rangle \equiv 0 \pmod q \}$.

\begin{lemma}\label{lemma:finalYes} %
If the $\CVP$ instance is a $\YES$ instance and the vector $s$
has the property from \expref{Lemma}{lemma:L_BCH}, then with probability at
least $99/100$ over the choice of the vector $w$, there exists a
lattice vector in $\calL(B)$ with $\ell_2$ norm at most $\gamma
\cdot \sqrt{d}$.
\end{lemma}

\begin{proof}
If the $\CVP$ instance is a $\YES$ instance and $s$
has the property from \expref{Lemma}{lemma:L_BCH}, then by \expref{Lemma}{lemma:Good}
there are at least $G$ good vectors in $\calL(B_{\tint})$, \ie, vectors with $\ell_2$
norm at most $\gamma \cdot \sqrt{d}$, coordinates from $\{0,1,2\}$, and
at least one coordinate equal to $1$.
For each good vector $x$, consider the event that $\langle w, x\rangle \equiv 0 \mbox{ (mod $q$)}$.
Since a good vector is nonzero, we clearly have that each such event occurs with probability $1/q$.
Moreover, observe that these vectors are pairwise linearly independent modulo $q$ and therefore
these events are pairwise independent.
Therefore, using Chebyshev's
Inequality, with probability at least $1-{q}/{G} \geq 99/100$, at least one of these events
happens, and we are done.
\end{proof}

\begin{lemma}\label{lemma:finalNo}
If the $\CVP$ instance is a $\NO$ instance, then with probability at
least $99/100$ over the choice of the vector $w$, for every
nonzero lattice vector $v \in \calL(B)$,
\begin{itemize}
    \item $v$ has at least $d$ nonzero coordinates, or
    \item all coordinates of $v$ are even and at least $d/4$ of them are nonzero, or
    \item all coordinates of $v$ are even and $\|v\|_2 \geq d$.
\end{itemize}
\end{lemma}

\begin{proof}
The probability that a nonzero lattice vector $x \in \calL(B_{\tint})$
satisfies $\langle w, x\rangle \equiv 0 \pmod q$ is
$1/q$. By the union bound, the probability that at least one
of the annoying vectors of $\calL(B_{\tint})$ belongs to $\calL(B)$ is
at most $A/q \leq 1/100$. Therefore, with
probability at least $99/100$, no lattice vector in $\calL(B)$ is annoying,
and the lemma follows.
\end{proof}

Lemmas~\ref{lemma:finalYes} and~\ref{lemma:finalNo} imply
\expref{Theorem}{thm:Khot}. \qed

\subsection*{Acknowledgements}
We thank Daniele Micciancio and Mario Szegedy for useful comments. We also
thank an anonymous referee for suggesting to avoid the use of huge
coordinates in Khot's proof, which in turn made \expref{Claim}{claim:sublattice}
simpler.

\bibliographystyle{tocplain}   %
\bibliography{v008a023}

\begin{tocauthors}
\begin{tocinfo}[ishay]
 Ishay Haviv\\
 School of Computer Science, The Academic College of Tel Aviv---Yaffo\\
 Tel Aviv, Israel \\[1ex]
 havivish\tocat{}tau\tocdot{}ac\tocdot{}il \\
\end{tocinfo}
\begin{tocinfo}[oded]
  Oded Regev\\
  Professor\\
  Blavatnik School of Computer Science, Tel Aviv University,\\[1ex]
   \quad and\\[1ex]
  CNRS, ENS Paris\\[1ex]
  regev\tocat{}di\tocdot{}ens\tocdot{}fr \\
  \url{http://www.cs.tau.ac.il/~odedr}
\end{tocinfo}
\end{tocauthors}

\begin{tocaboutauthors}
\begin{tocabout}[ishay]
\textsc{Ishay Haviv}
 graduated from \href{http://www.tau.ac.il/}{Tel Aviv University} in 2011
 under the supervision of \href{http://www.cs.tau.ac.il/~odedr/}{Oded
 Regev}. His research interests include computational aspects of
 lattices, coding theory, and other topics in theoretical computer science.
\end{tocabout}
\begin{tocabout}[oded]
\textsc{Oded Regev} graduated from
  \href{http://www.tau.ac.il/}{Tel Aviv University} in 2001 under the
  supervision of \href{http://www.cs.tau.ac.il/~azar/}{Yossi Azar}.
  He spent two years as a postdoc at the \href{http://www.ias.edu/}{Institute for
    Advanced Study}, Princeton, and one year at the
  \href{http://www.berkeley.edu/}{University of California, Berkeley}.
  He is currently with the
  \href{http://www.di.ens.fr}{cryptography group}
  at the
  \href{http://www.di.ens.fr}{{\'E}cole Normale Sup{\'e}rieure},
  Paris.
  His research interests include quantum computation, computational
  aspects of lattices, and other topics in theoretical computer
  science. He also enjoys photography, especially of his baby girl.
\end{tocabout}
\end{tocaboutauthors}

\end{document}